\documentclass[12pt]{iopart}
\usepackage{amsfonts}
\usepackage{geometry}
\usepackage{graphicx}

\newtheorem{theorem}{Theorem}
\newtheorem{axiom}{Axiom}
\newtheorem{definition}{Definition}
\newtheorem{corollary}{Corollary}
\newtheorem{remark}{Remark}
\newenvironment{proof}[1][Proof]{\noindent\textbf{#1.} }{\ \rule{0.5em}{0.5em}}
\newenvironment{observation}[1][Commentary]{\noindent\textbf{#1.} }{\ \rule{0.5em}{0.5em}}
\newtheorem{example}{Example}

\begin{document}

\title{Riemannian thermo-statistics geometry}
\author{L. Velazquez}

\begin{abstract}
It is developed a Riemannian reformulation of classical statistical
mechanics for systems in thermodynamic equilibrium, which arises as
a natural extension of Ruppeiner geometry of thermodynamics. The
present proposal leads to interpret entropy
$\mathcal{S}_{g}(I|\theta)$ and all its associated
thermo-statistical quantities as purely geometric notions derived
from the Riemannian structure on the manifold of macroscopic
observables $\mathcal{M}_{\theta}$ (existence of a distance
$ds^{2}=g_{ij}(I|\theta)dI^{i}dI^{j}$ between macroscopic
configurations $I$ and $I+dI$). Moreover, the concept of statistical
curvature scalar $R(I|\theta)$ arises as an invariant measure to
characterize the existence of an \textit{irreducible statistical
dependence} among the macroscopic observables $I$ for a given value
of control parameters $\theta$. This feature evidences a certain
analogy with Einstein General Relativity, where the spacetime
curvature $R(\mathbf{r},t)$ distinguishes the geometric nature of
gravitation and the reducible character inertial forces with an
appropriate selection of the reference frame.
\end{abstract}
\address{Departamento de F\'\i sica, Universidad Cat\'olica del Norte, Av. Angamos 0610, Antofagasta,
Chile.} \tableofcontents

\section{Introduction}

Let be a classical equilibrium distribution function:
\begin{equation}  \label{dist}
dp(I|\theta)=\rho(I|\theta)dI,
\end{equation}
where the stochastic variables $I=\left\{I^{i}\right\}$ denote the
relevant macroscopic observables of the system under analysis, while
$\theta=\left\{\theta^{\alpha}\right\}$ the underlying control
parameters of a given equilibrium situation. It is possible to
introduce \textit{a distance notion} between two close equilibrium
situations $\theta$ and $\theta+d\theta$:
\begin{equation}
ds^{2}_{F}=g_{\alpha\beta}(\theta)d\theta^{\alpha}d\theta^{\beta}
\end{equation}
in terms of the so-called Fisher's information matrix \cite{Fisher}:
\begin{equation}
g_{\alpha\beta}(\theta)=\int_{\mathcal{M}_{\theta}} \frac{%
\partial\log\rho(I|\theta)}{\partial\theta^{\alpha}} \frac{%
\partial\log\rho(I|\theta)}{\partial\theta^{\beta}}\rho(I|\theta)dI.
\end{equation}
The existence of this type of Riemannian formulation was pioneering
proposed by Rao \cite{Rao}, which is referred to as
\textit{inference geometry} in the literature \cite{Amari}.
Alternately, it is also possible to introduce a distance notion
between two close macroscopic configurations $I$ and $I+dI$ for a
given value of control parameter $\theta$:
\begin{equation}\label{dist2}
ds^{2}=g_{ij}(I|\theta)dI^{i}dI^{j}
\end{equation}
starting from the same distribution function (\ref{dist}). This
latter geometric characterization appears as a suitable extension of
Ruppeiner geometry of thermodynamics \cite{Ruppeiner1} to the
framework of equilibrium classical statistical mechanics. For this
reason, the present proposal shall be hereafter referred to as
\textit{Riemannian thermo-statistics geometry}. The main interest of
this work is to present a systematic development of the most
relevant features of this geometric formulation starting from an
\textit{axiomatic perspective}. Such a procedure allows to show that
fundamental concepts of statistical mechanics can be suitably
rephrased in terms of geometry notions, which provides a general
framework to apply powerful methods of Riemannian geometry for the
analysis of properties of thermodynamical systems.

\section{Riemannian thermo-statistics geometry}

Let us denote by $\mathcal{M}$ the abstract manifold composed of all
admissible values of macroscopic observables $I$, as well as by
$\mathcal{P}$ the abstract manifold composed of all admissible
values of control parameters $\theta$. Besides, let us also denote
by $\mathcal{M}_{\theta}$ the sub-manifold of $\mathcal{M}$ composed
of all values of macroscopic observables $I$ that are accessible for
a given value $\theta$. In general, it is possible to consider two
different types of \textit{coordinate reparametrizations}: (1) the
coordinate reparametrizations $\Theta(I):
\mathcal{R}_{I}\rightarrow\mathcal{R}_{\Theta}$ of the manifold of
macroscopic observables $\mathcal{M}_{\theta}$, as well as (2) the
coordinate reparametrizations
$\nu(\theta):\mathcal{R}_{\theta}\rightarrow\mathcal{R} _{\nu}$ of
the manifold of control parameters $\mathcal{P}$.

\subsection{Postulates of thermo-statistics geometry}

\begin{axiom}
The manifold of the system macroscopic observables
$\mathcal{M}_{\theta}$ possesses a \textbf{Riemannian structure},
that is, it is provided of a \textbf{metric tensor} $g_{ij}\left(
I|\theta \right)$ and a \textbf{torsionless  covariant
differentiation} $D_{i}$ that obeys the following constraints:
\begin{equation}
D_{k}g_{ij}\left( I|\theta \right) =0.  \label{Dg}
\end{equation}
\end{axiom}
\begin{definition}
The Riemannian structure on the manifold of macroscopic observables
$\mathcal{M}_{\theta}$ allows to introduce the \textbf{invariant
volume element} as follows:
\begin{equation}  \label{inv.volume}
d\mu(I|\theta)=\sqrt{\left\vert \frac{g_{ij}\left( I|\theta \right) }{2\pi }%
\right\vert }dI,
\end{equation}
where $\left\vert g_{ij}\left( I|\theta \right)\right\vert $ denotes
the absolute value of the metric tensor determinant.
\end{definition}
\begin{axiom}
There exist a differentiable scalar function $\mathcal{S}_{g}\left(
I|\theta \right) $ defined on the manifold $\mathcal{M}_{\theta }$,
hereafter referred to as the \textbf{scalar entropy}, whose
knowledge determines the equilibrium distribution function $dp\left(
I|\theta \right)$ of the macroscopic observables
$I\in\mathcal{M}_{\theta}$ as follows:
\begin{equation}
dp\left( I|\theta \right) =\exp \left[ \mathcal{S}_{g}\left(
I|\theta \right) \right] d\mu(I|\theta).  \label{EinsteinPostulate}
\end{equation}
\end{axiom}
\begin{definition}
Let us consider an arbitrary curve given in parametric form
$I(t)\in\mathcal{M}_{\theta}$ with fixed extreme points $I(t_{1})=P$
and $I(t_{2})=Q$. Adopting the following notation:
\begin{equation}
\dot{I}^{i}=\frac{d I^{i}(t)}{dt},
\end{equation}
the \textbf{length} $\Delta s $ of this curve can be expressed as:
\begin{equation}\label{distance}
\Delta s=\int_{t_{1}}^{t_{2}}\sqrt{g_{ij}\left[ I(t)|\theta \right] \dot{I}%
^{i}\left( t\right) \dot{I}^{j}\left( t\right) }dt.
\end{equation}
\end{definition}
\begin{definition}
The curve $I(t)\in\mathcal{M}_{\theta}$ exhibits an \textbf{unitary
affine parametrization} when its parameter $t$ satisfies the
following constraint:
\begin{equation}\label{affine}
g_{ij}(I|\theta)\dot{I}^{i}(t)\dot{I}^{j}(t)=1.
\end{equation}
\end{definition}
\begin{definition}
A \textbf{geodesic} is the curve $I_{g}\left( t\right) $ with
minimal length (\ref{distance}) between two fixed arbitrary points
$(P,Q)\in\mathcal{M}_{\theta}$. Moreover, the \textbf{distance}
$D(P,Q|\theta)$ between these two points $(P,Q)$ is given by the
length of its associated geodesic $I_{g}(t)$:
\begin{equation}\label{distance}
D(P,Q|\theta)=\int_{t_{1}}^{t_{2}}\sqrt{g_{ij}\left[ I_{g}(t)|\theta \right] \dot{I}_{g}%
^{i}\left( t\right) \dot{I}_{g}^{j}\left( t\right) }dt.
\end{equation}
\end{definition}
\begin{definition}
Considering a differentiable curve $I(t)\in\mathcal{M}_{\theta}$
with an unitary affine parametrization, the \textbf{entropy
production} along this curve $\Phi(t)$ is given by:
\begin{equation}\label{rate}
\Phi(t)=\frac{d\mathcal{S}_{g}\left[I(t)|\theta\right]}{dt}.
\end{equation}
\end{definition}
\begin{axiom}
The length $\Delta s$ of any interval $(t_{1},t_{2})$ of an
arbitrary \textit{geodesic} $I_{g}(t)\in\mathcal{M}_{\theta}$ with
an unitary affine parametrization is given by the variation of its
entropy production $\Delta\Phi(t)$ with opposite sight:
\begin{equation}\label{metric}
\Delta s=-\Delta\Phi(t)=\Phi(t_{1})-\Phi(t_{2}).
\end{equation}
\end{axiom}

\begin{axiom}
The probability density $\rho(I|\theta)$ associated with
distribution function (\ref{EinsteinPostulate}) vanishes with its
first partial derivatives for any point on the boundary
$\partial\mathcal{M}_{\theta}$ of the manifold
$\mathcal{M}_{\theta}$.
\end{axiom}

\subsection{Interpretations and fundamental consequences}
\begin{observation}
\textbf{Axiom 1} allows to precise the Riemannian structure of the
manifold $\mathcal{M}_{\theta}$ starting from the knowledge of the
metric tensor $ g_{ij}(I|\theta)$, specifically, the covariant
differentiation $D_{i}$ and the curvature tensor
$R_{ijkl}(I|\theta)$. As discussed elsewhere \cite{Berger},
Eq.(\ref{Dg}) is an strong constraint of Riemannian geometry that
determines the \textit{affine connections} $\Gamma _{ij}^{k}$
employed to introduce the covariant differentiation $D_{i}$,
specifically, the so-called \textit{Levi-Civita connection}:
\begin{equation}
\Gamma _{ij}^{k}\left( I|\theta \right) =g^{km}\frac{1}{2}\left( \frac{%
\partial g_{im}}{\partial I^{j}}+\frac{\partial g_{jm}}{\partial I^{i}}-%
\frac{\partial g_{ij}}{\partial I^{m}}\right).  \label{Levi-Civita}
\end{equation}
The knowledge of the affine connections $\Gamma _{ij}^{k}$ allows
the introduction of the \textit{curvature tensor} $
R^{l}_{ijk}=R^{l}_{ijk}(I|\theta)$ of the manifold
$\mathcal{M}_{\theta}$:
\begin{equation}\label{curvature}
R^{l}_{ijk}=\frac{\partial}{\partial I^{i}}\Gamma^{l}_{jk}-\frac{\partial}{%
\partial I^{j}}\Gamma^{l}_{ik}+\Gamma^{l}_{im}\Gamma^{m}_{jk}-
\Gamma^{l}_{jm}\Gamma^{m}_{ik},
\end{equation}
which is also determined from the knowledge of the metric tensor $
g_{ij}(I|\theta)$ and its first and second partial derivatives.
Using the curvature tensor $R^{l}_{ijk}(I|\theta)$, it is possible
to obtain its fourth-rank covariant form
$R_{ijkl}(I|\theta)=g_{lm}(I|\theta)R^{m}_{ijk}(I|\theta)$:
\begin{eqnarray}\label{curvature.2}\nonumber
&&R_{ijkl}=\frac{1}{2}\left(\frac{\partial^{2}g_{il}}{\partial
I^{j}\partial I^{k}}+\frac{\partial^{2}g_{jk}}{\partial
I^{i}\partial I^{l}}-\frac{\partial^{2}g_{jl}}{\partial
I^{i}\partial I^{k}}-\frac{\partial^{2}g_{ik}}{\partial
I^{j}\partial I^{l}}\right)+\\
&&+g_{mn}\left(\Gamma^{m}_{il}\Gamma^{n}_{jk}-\Gamma^{m}_{jl}\Gamma^{n}_{ik}\right),
\end{eqnarray}
the \textit{Ricci curvature tensor} $R_{ij}(I|\theta)$:
\begin{equation}
R_{ij}(I|\theta)=R_{kij}^{k}(I|\theta)
\end{equation}
as well as the \textit{curvature scalar} $R(I|\theta)$:
\begin{equation} \label{scalar.curvature}
R(I|\theta)=g^{ij}(I|\theta)R^{k}_{kij}(I|\theta)=g^{ij}(I|\theta)g^{kl}(I|\theta)R_{kijl}(I|\theta).
\end{equation}
The curvature scalar $R(I|\theta)$ has a paramount relevance in
Riemannian geometry \cite{Berger} because of it is the only
invariant derived from the first and second partial derivatives of
the metric tensor $g_{ij}(I|\theta)$.
\end{observation}

\begin{observation}
\textbf{Axiom 2} is a \textit{covariant redefinition of Einstein
postulate} of classical fluctuation theory \cite{landau}:
\begin{equation}\label{EP}
dp_{EP}(I|\theta)=\mathcal{A}\exp\left[S(I|\theta)\right]dI,
\end{equation}
which rephrases Boltzmann entropy $S=\log W$ to assign
\textit{relative probabilities} from the entropy $S(I|\theta)$.
However, expression (\ref{EP}) has the disadvantage that the entropy
$S(I|\theta)$, commonly referred to as the \textit{coarsed grained
entropy}, does not correspond to a \textit{scalar function} within a
geometric theory. In fact, coarsed grained entropy $S(I|\theta)$
behaves under coordinate reparametrizations
$\Theta(I):\mathcal{R}_{I}\rightarrow\mathcal{R}_{\Theta}$ as:
\begin{equation}
S(\Theta|\theta)=S(I|\theta)-\log\left|\frac{\partial\Theta}{\partial
I}\right|.
\end{equation}
The non-scalar character of the entropy $S\left( I|\theta \right)$
is inconsistent with the physical idea that \textit{thermodynamic
entropy is a state function}, which should not depend on the
particular coordinate representation employed to describe the
macroscopic properties of a thermodynamic system. One can avoid this
mathematical inconsistence replacing the usual volume element $dI$
by the invariant volume element (\ref{inv.volume}). Thus, the scalar
character of the entropy $ \mathcal{S}_{g}\left( I|\theta \right)$:
\begin{equation}
\mathcal{S}_{g}\left( \Theta|\theta \right)=\mathcal{S}_{g}\left(
I|\theta \right)
\end{equation}
and the covariance of the metric tensor $g_{ij}\left( I|\theta
\right)$:
\begin{equation}
g_{ij}\left( \Theta|\theta \right)=\frac{\partial I^{m}}{\partial\Theta^{i}}%
\frac{\partial I^{n}}{\partial\Theta^{j}}g_{mn}\left( I|\theta
\right)
\end{equation} guarantee the
invariance of the equilibrium distribution function
(\ref{EinsteinPostulate}) under coordinate reparametrizations
$\Theta(I):\mathcal{R}_{I}\rightarrow \mathcal{R}_{\Theta}$.
\end{observation}

\begin{remark}[Alternative \textbf{Axiom 2}]
Starting from the knowledge of probability density $\rho \left(
I|\theta \right)$ of the distribution function (\ref{dist}) and the
existence of an everywhere non-vanishing metric determinant
$\left|g_{ij}\left( I|\theta \right)\right| $, it is possible to
introduce the scalar entropy $\mathcal{S}_{g}\left( I|\theta \right)
$ as follows:
\begin{equation}
\mathcal{S}_{g}\left( I|\theta \right) \equiv \log \frac{\rho \left(
I|\theta \right) }{\sqrt{\left\vert g_{ij}\left( I|\theta \right)
/2\pi \right\vert }}.  \label{ScalarEntropy}
\end{equation}
\end{remark}

\begin{observation}
\textbf{Axiom 3} states a direct relation between the distance
notion (\ref{dist2}) and the entropy production (\ref{rate}), which
relates the metric tensor $g_{ij}(I|\theta)$ and the scalar entropy
$\mathcal{S}_{g}(I|\theta)$. This postulate can be regarded as a
generalization of the Ruppeiner postulate \cite{Ruppeiner1}:
\begin{equation}
g_{ij}\left( \bar{I}|\theta \right) =-\frac{\partial
^{2}\mathcal{S}_{g}\left( \bar{I}|\theta \right) }{\partial
I^{i}\partial I^{k}} \label{RuppeinerAnsatz}
\end{equation}
which provides a metric tensor for thermodynamics, with $\bar{I}$
being the point with global maximum entropy.
\end{observation}
\begin{theorem}
\textit{The scalar entropy} $\mathcal{S}_{g}(I|\theta)$ \textit{is
locally concave everywhere} and the metric tensor $g_{ij}(I|\theta)$
is \textit{positive definite} on the manifold
$\mathcal{M}_{\theta}$. Moreover, the metric tensor
$g_{ij}(I|\theta)$ can be identified with the \textbf{covariant
Hessian} $\mathcal{H}_{ij}(I|\theta)$ of the scalar entropy
$\mathcal{S}_{g}\left( I|\theta \right)$ with opposite sign:
\begin{equation}
g_{ij}\left( I|\theta \right)
=-\mathcal{H}_{ij}(I|\theta)=-D_{i}D_{j}\mathcal{S}_{g}\left(
I|\theta \right) . \label{cov.EH}
\end{equation}
\end{theorem}
\begin{proof}
The searching of the curve with minimal length (\ref{distance})
between two arbitrary points $(P,Q)$ is a variational problem that
leads to the following ordinary differential equations
\cite{Berger}:
\begin{equation} \label{geodesic}
\dot{I}_{g}^{k}(t)D_{k}\dot{I}_{g}^{i}(t)=\ddot{I}_{g}^{i}(t)+
\Gamma^{i}_{mn}\left[I_{g}(t)|\theta\right]\dot{I}_{g}^{m}(t)\dot{I}%
_{g}^{n}(t)=0,
\end{equation}
which describes the geodesic $I_{g}(t)$ with an unitary affine
parametrization. Eqs.(\ref{rate}) and (\ref{metric}) can be
rephrased as follows:
\begin{equation}\label{alter.post3}
\Delta s=-\Delta\Phi(s)\rightarrow\frac{d\Phi(s)}{ds}=
\frac{d^{2}\mathcal{S}_{g}}{ds^{2}}=-1.
\end{equation}
Taking into account the geodesic differential equations
(\ref{geodesic}), constraint (\ref{alter.post3}) can be rewritten
as:
\begin{eqnarray}  \label{razonamiento}
\frac{d^{2}\mathcal{S}_{g}}{ds^{2}} = \ddot{I}_{g}^{k}\frac{\partial
\mathcal{S}_{g}}{\partial I^{k}}+\dot{I}_{g}^{i}\dot{I}_{g}^{j}\frac{%
\partial^{2} \mathcal{S}_{g}}{\partial I^{i}\partial I^{j}}
=\dot{I}_{g}^{i}\dot{I}_{g}^{j}\left\{\frac{\partial^{2} \mathcal{S}_{g}}{%
\partial I^{i}\partial I^{j}}-\Gamma^{k}_{ij}\frac{\partial \mathcal{S}_{g}}{%
\partial I^{k}}\right\},
\end{eqnarray}
where it is possible to identify the \textit{covariant entropy
Hessian} $\mathcal{H}_{ij}$:
\begin{equation}\label{Hessian}
\mathcal{H}_{ij}=D_{i}D_{j}\mathcal{S}_{g}=\frac{\partial^{2} \mathcal{S}_{g}}{%
\partial I^{i}\partial I^{j}}-\Gamma^{k}_{ij}\frac{\partial \mathcal{S}_{g}}{%
\partial I^{k}}.
\end{equation}
Eqs.(\ref{razonamiento})-(\ref{Hessian}) can be combined with
constraint (\ref{affine}) to obtain the following expression:
\begin{equation}
(g_{ij}+\mathcal{H}_{ij})\dot{I}_{g}^{i}\dot{I}_{g}^{j}=0.
\end{equation}
Its covariant character leads to Eq.(\ref{cov.EH}). The concave
behavior of the scalar entropy $\mathcal{S}_{g}(I|\theta)$ and the
positive definition of the metric tensor $g_{ij}(I|\theta)$ are two
direct consequences of Eq.(\ref{alter.post3}).
\end{proof}
\begin{corollary}
The metric tensor $g_{ij}\left( I|\theta \right) $ can be obtained
from a given scalar entropy $\mathcal{S}_{g}\left( I|\theta \right)$
through the following set of \textit{covariant partial differential
equations}:
\begin{equation}  \label{cov.equation}
g_{ij}=-\frac{\partial^{2}\mathcal{S}_{g}}{\partial I^{i}\partial I^{j}}+%
\frac{1}{2}g^{km}\left( \frac{\partial g_{im}}{\partial I^{j}}+\frac{%
\partial g_{jm}}{\partial I^{i}}-\frac{\partial g_{ij}}{\partial I^{m}}%
\right)\frac{\partial\mathcal{S}_{g}}{\partial I^{k}}.
\end{equation}
The admissible solutions derived from the nonlinear problem
(\ref{cov.equation}) should be everywhere finite and
differentiable, including also on boundary of the manifold $\mathcal{M}%
_{\theta}$.
\end{corollary}
\begin{corollary}
Ruppeiner geometry of thermodynamics \cite{Ruppeiner1} is obtained
from as a particular case of thermo-statistics geometry restricting
to the gaussian approximation of the distribution function
(\ref{EinsteinPostulate}):
\begin{equation}  \label{Gaussian}
dp\left( I|\theta \right) \simeq \mathcal{N}\exp \left[ -\frac{1}{2}%
g_{ij}\left( \bar{I}|\theta \right) \delta I^{i}\delta I^{j}\right] \sqrt{%
\left\vert \frac{g_{ij}\left( \bar{I}|\theta \right) }{2\pi
}\right\vert }dI,
\end{equation}
where $\delta I=I-\bar{I}$ and
$\mathcal{N}=\exp\left[\mathcal{S}_{g}\left( \bar{I}|\theta
\right)\right]$. Moreover, the normalization of distribution
function (\ref{Gaussian}) leads to the following estimation for the
maximum scalar entropy $\mathcal{S}_{g}\left( \bar{I}|\theta \right)
$:
\begin{equation}
\mathcal{S}_{g}\left( \bar{I}|\theta \right) \simeq 0.
\label{estimation}
\end{equation}
\end{corollary}
\begin{proof}
Eq.(\ref{cov.EH}) drops to Ruppeiner definition of thermodynamic
metric tensor (\ref{RuppeinerAnsatz}) restricting to the point
$\bar{I}$ with of the global maximum scalar entropy
$\mathcal{S}_{g}(I|\theta)$. For a sufficiently large thermodynamic
short-range interacting system, it is possible to assume the
following scaling dependencies with the characteristic size
parameter $\Lambda$, $\mathcal{S}_{g}\sim \Lambda$, $I\sim
\Lambda\Rightarrow g_{ij}\sim1/\Lambda$ and the correlation
functions $\left\langle\delta I^{i}\delta I^{j}\right\rangle\sim
\Lambda$, where the role of $\Lambda$ can be performed by the system
volume $V$, the total mass $M$, number of constituents $N$, etc.
Gaussian distribution (\ref{Gaussian}) represents the lowest
approximation level that accounts for the thermodynamic fluctuations
of the macroscopic observables $I$ in the power expansion of the
scalar entropy $\mathcal{S}_{g}(I|\theta)$ and the logarithm of the
factor $\sqrt{|g_{ij}(I|\theta)|}$.
\end{proof}

\begin{observation}
\textbf{Axiom 4} talks about the asymptotic behavior of the
equilibrium distribution function (\ref{EinsteinPostulate}) for any
point $I_{b}$ on the boundary $\partial\mathcal{M}_{\theta}$:
\begin{equation}  \label{boundary.cond}
\lim_{I\rightarrow I_{b}}\rho(I|\theta)=\lim_{I\rightarrow I_{b}}\frac{%
\partial}{\partial I^{i}}\rho(I|\theta)=0.
\end{equation}
Such conditions play a fundamental role in the character of
stationary points (maxima and minima) of the scalar entropy
$\mathcal{S}_{g}(I|\theta)$ as well as the \textit{fluctuation
theorems} derived on the Riemannian structure of the manifold
$\mathcal{M}_{\theta}$. Due to their own importance, such
fluctuation theorems shall be discussed in a forthcoming paper.
\end{observation}

\begin{remark}
The boundary conditions (\ref{boundary.cond}) are independent from
the admissible coordinate representation $\mathcal{R}_{I}$ of the
manifold $\mathcal{M}_{\theta}$.
\end{remark}
\begin{proof}
This remark follows as a direct consequence of the transformation
rule of the probability density:
\begin{equation}
\rho (\Theta |\theta)=\rho (I|\theta)\left\vert \frac{\partial \Theta }{%
\partial I}\right\vert ^{-1}  \label{rho.tr}
\end{equation}
as well as the ones associated with its derivatives:
\begin{equation}  \label{drho.tr}
\frac{\partial \rho \left( \Theta |\theta \right) }{\partial \Theta ^{i}} =
\frac{\partial I^{j}}{\partial \Theta ^{i}}\left\{ \frac{\partial \rho
\left( I|\theta \right) }{\partial I^{j}}-\rho \left( I|\theta \right) \frac{%
\partial }{\partial I^{j}}\log \left\vert \frac{\partial \Theta }{\partial I}%
\right\vert \right\} \left\vert \frac{\partial \Theta }{\partial I}%
\right\vert ^{-1}
\end{equation}
under a coordinate reparametrization
$\Theta(I):\mathcal{R}_{I}\rightarrow\mathcal{R}_{\Theta}$ whose
Jacobian $\left|\partial\Theta/\partial I \right|$ be everywhere
finite and differentiable.
\end{proof}

\section{Gaussian and spherical representations}
\subsection{Gaussian Planck potential $\mathcal{P}_{g}(\theta)$}
\begin{definition}
The covariant form of the \textbf{gradiental generalized forces}
$\psi _{i}\left( I|\theta \right) $ are defined from the scalar
entropy $\mathcal{S}_{g}\left( I|\theta \right)$ as follows:
\begin{equation}
\psi _{i}\left( I|\theta \right) =-D_{i}\mathcal{S}_{g}\left(
I|\theta \right) \equiv-\partial \mathcal{S}_{g}\left( I|\theta
\right) /\partial I^{i}.  \label{cov.DGF}
\end{equation}
Using the metric tensor $g^{ij}(I|\theta)$, it is possible to obtain
its contravariant counterpart $\psi^{i}(I|\theta)$:
\begin{equation}
\psi ^{i}\left( I|\theta \right) =g^{ij}\left( I|\theta \right) \psi
_{j}\left( I|\theta \right),
\end{equation}
as well as its the square norm $\psi ^{2}=\psi ^{2}(I|\theta)$:
\begin{equation}
\psi ^{2}(I|\theta)=\psi ^{i}\left( I|\theta \right)\psi _{i}\left(
I|\theta \right).
\end{equation}
\end{definition}

\begin{theorem}
The scalar entropy $\mathcal{S}_{g}(I|\theta )$ can be expressed in
terms of the square norm of the gradiental generalized forces as
follows:
\begin{equation}
\mathcal{S}_{g}(I|\theta )=\mathcal{P}_{g}(\theta )-\frac{1}{2}%
\psi^{2}(I|\theta ), \label{Sdecomposition}
\end{equation}
where $\mathcal{P}_{g}(\theta )$ is a certain function on control
parameters $\theta$ hereafter referred to as the \textbf{gaussian
Planck potential}.
\end{theorem}
\begin{proof}
Let us introduce the scalar function $\mathcal{P}_{g}(I|\theta )$:
\begin{equation}
\mathcal{P}_{g}(I|\theta )=\mathcal{S}_{g}(I|\theta )+\frac{1}{2}%
g^{ij}(I|\theta )\psi _{i}(I|\theta )\psi _{j}(I|\theta ).
\end{equation}%
It is easy to verify that its covariant derivatives:
\begin{eqnarray}
&&D_{k}\mathcal{P}_{g}(I|\theta )=D_{k}S_{g}(I|\theta )+
\frac{1}{2}\left\{ \psi _{i}(I|\theta )\psi _{j}(I|\theta
)D_{k}g^{ij}(I|\theta )+\right.   \\
&&\left. +g^{ij}(I|\theta )\left[ \psi _{i}(I|\theta )D_{k}\psi
_{j}(I|\theta )+\psi _{j}(I|\theta )D_{k}\psi _{i}(I|\theta )\right]
\right\} \nonumber
\end{eqnarray}
vanish as direct consequences of the metric tensor properties
(\ref{Dg}) and (\ref{cov.EH}), as well as the definition
(\ref{cov.DGF}) of the gradiental generalized forces $\psi
_{i}(I|\theta )$. Since the covariant derivatives of any scalar
function are given by the usual partial derivatives:
\begin{equation}
D_{k}\mathcal{P}_{g}(I|\theta )=\frac{\partial }{\partial I^{k}}\mathcal{P}%
_{g}(I|\theta )=0,
\end{equation}%
the scalar function $\mathcal{P}_{g}(I|\theta )$ only depends on the control
parameters:
\begin{equation}
\mathcal{P}_{g}(I|\theta )\equiv \mathcal{P}_{g}(\theta ).
\end{equation}
This last result leads to Eq.(\ref{Sdecomposition}).
\end{proof}
\begin{corollary}\label{entropy.planck}
The value of scalar entropy $\mathcal{S}_{g}(I|\theta )$ at all its
extreme points derived from the stationary condition:
\begin{equation}\label{global.stat}
\psi^{2}(\bar{I}|\theta)=0
\end{equation}
is exactly given by the gaussian Planck potential
$\mathcal{P}_{g}(\theta)$.
\end{corollary}
\begin{corollary}
The equilibrium distribution function (\ref{EinsteinPostulate})
admits the following \textbf{gaussian representation}:
\begin{equation}
dp(I|\theta )=\frac{1}{\mathcal{Z}_{g}(\theta )}\exp \left[ -\frac{1}{2}%
\psi^{2} (I|\theta )\right] d\mu (I|\theta ). \label{universal}
\end{equation}
Here, the factor $\mathcal{Z}_{g}(\theta )$ is related to the
gaussian Planck potential as follows:
\begin{equation}\label{gaussianPlanck}
\mathcal{P}_{g}(\theta )=-\log \mathcal{Z}_{g}(\theta ),
\end{equation}
which shall be hereafter referred to as the \textbf{gaussian
partition function}.
\end{corollary}

\subsection{Maximum and completeness theorems}
\begin{theorem}
The scalar entropy $\mathcal{S}_{g}(I|\theta)$ exhibits a unique
stationary point $\bar{I}$, which corresponds to its global maximum.
\end{theorem}
\begin{proof}
Since the metric tensor $g_{ij}(I|\theta)$ defines a positive
definite metric on the manifold $\mathcal{M}_{\theta}$, the
vanishing of the probability density:
\begin{equation}
\rho(I|\theta)=\exp\left[\mathcal{S}_{g}(I|\theta)\right]\sqrt{\left\vert \frac{g_{ij}\left( I|\theta \right) }{2\pi }%
\right\vert }
\end{equation}
on the boundary $\partial\mathcal{M}_{\theta}$ of the manifold
$\mathcal{M}_{\theta}$ for every admissible coordinate
representation $\mathcal{R}_{I}$ evidences the vanishing of the
scalar function:
\begin{equation}
\Psi(I|\theta)=\exp\left[\mathcal{S}_{g}(I|\theta)\right]
\end{equation}
on the boundary $\partial\mathcal{M}_{\theta}$. Since the function
$\Psi(I|\theta)$ is nonnegative, finite and differentiable on the
manifold $\mathcal{M}_{\theta}$, the entropy
$\mathcal{S}_{g}(I\theta)$ should exhibit at least a stationary
point where takes place the stationary condition
(\ref{global.stat}). Since the scalar entropy
$\mathcal{S}_{g}(I\theta)$ is a concave function, its stationary
points can only correspond to local maxima. Let us suppose the
existence of a least two stationary points $\bar{I}_{1}$ and
$\bar{I}_{2}$, which can always be connected with a certain geodesic
$I_{g}(t)$. According to constraint (\ref{alter.post3}), the entropy
production $\Phi(t)$ is a monotonous function along the curve
$I_{g}(t)$. Consequently, $\Phi(t)$ should exhibits different values
at the stationary points $\bar{I}_{1}$ and $\bar{I}_{2}$, which is
absurdum since the entropy production $\Phi(t)$ identically vanishes
for any stationary point of the scalar entropy
$\mathcal{S}_{g}(I|\theta)$:
\begin{equation}
\Phi(t)=-\dot{I}^{i}(t)\psi_{i}[I(t)|\theta].
\end{equation}
Consequently, there exist only one stationary point that corresponds
with the global maximum of the scalar entropy
$\mathcal{S}_{g}(I|\theta)$.
\end{proof}
\begin{theorem}\label{entropy.distance}
Any hyper-surface of constant scalar entropy
$\mathcal{S}_{g}(I|\theta)$ is just the boundary of a
$n$-dimensional sphere $S^{n}(\bar{I},r)\subset
\mathcal{M}_{\theta}$ centered at the point $\bar{I}$ with global
maximum entropy, whose radio $r=r(I|\theta)=D(I,\bar{I}|\theta)$ at
the point $I$ allows to rephrase the scalar entropy as follows:
\begin{equation}\label{E1}
\mathcal{S}_{g}(I|\theta)=\mathcal{P}_{g}(\theta)-\frac{1}{2}r^{2}(I|\theta),
\end{equation}
with $n=\dim\mathcal{M}_{\theta}$ being the dimension of the
manifold $\mathcal{M}_{\theta}$.
\end{theorem}
\begin{proof}
By definition, the vector field $\upsilon^{i}(I|\theta)$:
\begin{equation}\label{unitary}
\upsilon^{i}(I|\theta)=\frac{\psi^{i}(I|\theta)}{\psi(I|\theta)}
\end{equation}
is the unitary normal vector of the hyper-surface with constant
scalar entropy $\mathcal{S}_{g}(I|\theta)$. It is easy to verify
that the vector field $\upsilon^{i}(I|\theta)$ obeys the geodesic
equations (\ref{geodesic}):
\begin{equation}
\upsilon^{k}(I|\theta)D_{k}\upsilon^{i}(I|\theta)=\frac{\upsilon^{k}
(I|\theta)}{\psi(I|\theta)}\left[\delta^{i}_{k}-\upsilon^{i}(I|\theta)\upsilon_{k}(I|\theta)\right]=0.
\end{equation}
Hence, $\upsilon^{i}(I|\theta)$ can be regarded as the tangent
vector:
\begin{equation}\label{velocity}
\frac{dI^{i}_{g}(s|\mathbf{e})}{ds}=\upsilon^{i}[I_{g}(s|\mathbf{e})|\theta]
\end{equation}
of geodesic family $I_{g}(s|\mathbf{e})$ with unitary affine
parametrization centered at the point $\bar{I}$ with maximum scalar
entropy $\mathcal{S}_{g}(I|\theta)$,
$I_{g}(s=0|\mathbf{e})=\bar{I}$, where the parameters $\mathbf{e}$
distinguish geodesics with different directions at the origin. The
entropy production $\Phi(s|\mathbf{e})$ along any of these geodesics
is given by the norm of the gradiental generalized forces with
opposite sign:
\begin{equation}
\Phi(s|\mathbf{e})=-\frac{dI^{i}(s|\mathbf{e})}{ds}\psi_{i}\left[I_{g}(s|\mathbf{e})|\theta\right]=-\psi[I_{g}(s|\mathbf{e})|\theta].
\end{equation}
Considering Eq.(\ref{metric}), the norm $\psi(I|\theta)$ can be
related to the length $\Delta s$ of the geodesic connecting the
point $I$ with point $\bar{I}$ with maximum scalar entropy, that is,
the \textit{distance} $D(I,\bar{I}|\theta)$ between the points $I$
and $\bar{I}$:
\begin{equation}\label{radio}
\psi(I|\theta)=D(I,\bar{I}|\theta).
\end{equation}
According to the gaussian decomposition (\ref{Sdecomposition}), the
hyper-surface with constant scalar entropy
$\mathcal{S}_{g}(I|\theta)$ is also the hyper-surface where the norm
of gradiental generalized forces $\psi(I|\theta)$ is kept constant,
that is, the boundary of a $n$-dimensional sphere $S^{n}(\bar{I},r)$
centered at the point $\bar{I}$ with maximum entropy.
\end{proof}
\begin{corollary}[Completeness]
The knowledge of the metric tensor $g_{ij}(I|\theta)$ and the point
$\bar{I}$ with maximum scalar entropy $\mathcal{S}_{g}(I|\theta)$
determines the equilibrium distribution function
(\ref{EinsteinPostulate}).
\end{corollary}
\begin{observation}
The radio $r(I|\theta)$ of the $n$-dimensional sphere
$S^{n}(\bar{I},r)$ of the \textbf{theorem \ref{entropy.distance}}
and the invariant volume element $d\mu(I|\theta)$ are purely
geometric notions derived from the knowledge of the metric tensor
$g_{ij}(I|\theta)$ and the point $\bar{I}$  with maximum scalar
entropy $\mathcal{S}_{g}(I|\theta)$. Therefore, \textit{the scalar
entropy $\mathcal{S}_{g}(I|\theta)$ and all its associated
thermo-statistical quantities represents geometric notions derived
from the Riemannian structure of the manifold
$\mathcal{M}_{\theta}$}.
\end{observation}
\begin{corollary}
The equilibrium distribution function (\ref{EinsteinPostulate}) can
be expressed as follows:
\begin{equation}\label{polar}
dp(r,q|\theta)=\frac{1}{\mathcal{Z}_{g}(\theta)}\frac{1}{\sqrt{2\pi}}\exp\left(-\frac{1}{2}r^{2}\right)
drd\Sigma_{g}(q|r,\theta)
\end{equation}
where $d\Sigma(q|r,\theta)$ is the hyper-surface element:
\begin{equation}
d\Sigma_{g}(q|r,\theta)=\sqrt{\left|\frac{g_{\alpha\beta}(r,q|\theta)}{2\pi}\right|}dq.
\end{equation}
obtained from the metric tensor $g_{\alpha\beta}(r,q|\theta)$
associated with the \textbf{projected Riemannian structure} on the
hyper-surface $\partial S^{(n)}(\bar{I}|r)$ of the n-dimensional
sphere $S^{(n)}(\bar{I}|r)\subset\mathcal{M}_{\theta}$.
\end{corollary}
\begin{proof}
Let us consider the geodesic family $I_{g}(s;\mathbf{e})$ derived as
a solution of the problem (\ref{velocity}). The quantities
$\mathbf{e}=\left\{e^{i}\right\}$ represent the asymptotic values of
the unitary vector field $\upsilon^{i}(I|\theta)$ at the origin:
\begin{equation}
e^{i}=\lim_{s\rightarrow 0}\dot{I}^{i}_{g}(s;\mathbf{e}),
\end{equation}
which characterize the direction of the given geodesic
$I_{g}(s;\mathbf{e})$ at the point $\bar{I}$. These vectors can be
parameterized as $\mathbf{e}=\mathbf{e}(q)$ using the intersection
point $q$ of the geodesic $I_{g}(s;\mathbf{e})$ with the boundary
$\partial S^{n}(\bar{I}|r)$ of a n-dimensional sphere
$S^{n}(\bar{I}|r)\subset\mathcal{M}_{\theta}$ of the \textbf{theorem
\ref{entropy.distance}}, which shall be hereafter referred to as the
\textbf{spherical coordinates}. One can employ the quantities
$\rho=(r,q)$ to introduce a \textbf{spherical representation}
$\mathcal{R}_{\rho}$ centered at the point $\bar{I}$ with maximum
scalar entropy $\mathcal{S}_{g}(I|\theta)$. The coordinate
reparametrization
$\rho(I):\mathcal{R}_{I}\rightarrow\mathcal{R}_{\rho}$ is given by
the geodesic family $I=I^{i}_{g}[r|\mathbf{e}(q)]$:
\begin{equation}
\upsilon^{i}(I|\theta)=\frac{\partial
I^{i}_{g}[r|\mathbf{e}(q)]}{\partial
r},\tau^{i}_{\alpha}(I|\theta)=\frac{\partial
I^{i}_{g}[r|\mathbf{e}(q)]}{\partial q^{\alpha}}.
\end{equation}
The new $n-1$ vector fields $\tau^{i}_{\alpha}(I|\theta)$ are
perpendicular to the unitary vector field $\upsilon_{i}(I|\theta)$,
$\upsilon_{i}(I|\theta)\tau^{i}_{\alpha}(I|\theta)=0$, since the
vectors $\upsilon^{i}(I|\theta)$ and $\tau^{i}_{\alpha}(I|\theta)$
are respectively normal and tangential to the boundary  $\partial
S^{(n)}(\bar{I}|r)$ of the n-dimensional sphere
$S^{(n)}(\bar{I}|r)$. Consequently, the components of the metric
tensor in this spherical coordinate representation are given by:
\begin{eqnarray}\nonumber
&g_{rr}(r,q|\theta)=1, g_{r\alpha}(r,q|\theta)=g_{\alpha
r}(r,q|\theta)=0,\\
&g_{\alpha\beta}(r,q|\theta)=g_{ij}(I|\theta)\tau^{i}_{\alpha}(I|\theta)\tau^{j}_{\beta}(I|\theta),
\end{eqnarray}
whose components $g_{\alpha\beta}(r,q|\theta)$ describe projected
Riemannian structure on the boundary $\partial S^{(n)}(\bar{I}|r)$
of the n-dimensional sphere $S^{(n)}(\bar{I}|r)$. Eq.(\ref{polar})
is just the expression of the distribution function
(\ref{EinsteinPostulate}) in the spherical representation
$\mathcal{R}_{\rho}$.
\end{proof}

\section{Statistical curvature}\label{CurvatureSect}

\begin{definition}
A set of stochastic variables $I$ exhibits a \textbf{reducible
statistical dependence} when there exists at least a coordinate
representation $\mathcal{R}_{I}$ of the manifold
$\mathcal{M}_{\theta}$ where the distribution function
$dp(I|\theta)$ can be factorized into independent distribution
functions
$dp^{(i)}(I^{(i)}|\theta)=\rho^{(i)}(I^{(i)}|\theta)dI^{i}$ as
follows:
\begin{equation}\label{factorization}
dp(I|\theta)=\prod^{n}_{i=1}dp^{(i)}(I^{(i)}|\theta)
\end{equation}
for any $I\in\mathcal{M}_{\theta}$. Conversely, the set of
stochastic variables $I$ exhibits an \textbf{irreducible statistical
dependence}.
\end{definition}
\begin{example}
The stochastic variables $X$ and $Y$ described by the distribution
function:
\begin{equation}\label{example.ind}
dp(X,Y)=A\exp\left[-X^{2}-Y^{2}-XY\right]dXdY
\end{equation}
are statistical dependent. However, this same distribution function
can be decomposed into independent distributions:
\begin{equation}
dp(\zeta,\eta)=\sqrt{\frac{3}{2\pi}}\exp\left[-\frac{3}{2}
\zeta^{2}\right]d\zeta\frac{1}{\sqrt{2\pi}}\exp\left[-\frac{1}{2}\eta^{2}\right]
d\eta
\end{equation}
considering the the following coordinate reparametrizations:
\begin{equation}
X=\frac{1}{\sqrt{2}}(\zeta+\eta),\:Y=\frac{1}{\sqrt{2}}(\zeta-\eta).
\end{equation}
This fact evidences that the distribution function
(\ref{example.ind}) exhibits a reducible statistical dependence.
Clearly, the correlation functions:
\begin{equation}\label{correlation}
cov(X,Y)=\left\langle(X-\left\langle X\right\rangle)(Y-\left\langle
Y\right\rangle)\right\rangle
\end{equation}
cannot provide an \textbf{absolute measure of statistical
dependence}.
\end{example}
\begin{definition}
A manifold $\mathcal{M}_{\theta}$ is \textbf{flat} when its
geometric properties are fully equivalent to the $n$-dimensional
Euclidean space $\mathbb{R}^{n}$. Otherwise, the manifold is said to
be \textbf{curved}.
\end{definition}
\begin{theorem}
The flat character of the manifold $\mathcal{M}_{\theta}$ implies
the existence of a reducible statistical dependence among
macroscopic observables $I$, while its curved character implies the
existence of an irreducible statistical dependence.
\end{theorem}
\begin{proof}
According to Riemannian geometry, the flat or curved character of
the manifold $\mathcal{M}_{\theta}$ is unambiguously determined by
the curvature tensor $R_{ijkl}(I|\theta)$. In fact, the manifold
$\mathcal{M}_{\theta}$ is flat whenever the curvature tensor
$R^{l}_{ijk}(I|\theta)$ vanishes for any $I\in\mathcal{M}_{\theta}$
\cite{Berger}. The metric tensor $g_{ij}(I|\theta)$ of an Euclidean
manifold $\mathcal{M}_{\theta}$ can be written in terms of delta
Kronecker matrix $\delta_{ij}$:
\begin{equation}
g_{ij}(I|\theta)=\delta_{ij}=-\partial^{2}\mathcal{S}%
_{g}(I|\theta)/\partial I^{i}\partial I^{j}
\end{equation}
using an appropriate coordinate representation $\mathcal{R}_{I}$.
Such a metric tensor leads to the following decomposition of the
equilibrium distribution function $dp(\Theta|\theta)$:
\begin{equation}  \label{GaussianEuclidean}
dp(I|\theta)=\prod_{i}\exp\left\{-\left[I^{i}-\bar{I}%
^{i}(\theta) \right]^{2}/2\right\}\frac{dI^{i}}{\sqrt{2\pi}}
\end{equation}
into gaussian distributions, which evidences the reducible character
of the statistical dependence among macroscopic observables $I$ for
the case of a flat manifold.

To show the relationship between curved character of the manifold
$\mathcal{M}_{\theta}$ and the irreducible statistical dependence
among the macroscopic observables $I$, let us consider the spherical
representation (\ref{polar}) of the distribution function
(\ref{EinsteinPostulate}). The demonstration reduces to verify the
irreducible coupling between the radial coordinate $r$ and the
spherical coordinates $q$ for any coordinate representation
$\mathcal{R}_{q}$ of the boundary $\partial S^{(n)}(\bar{I},r)$. The
analysis can be restricted to the asymptotic behavior of the
distribution functions for $r$ small, that is, within the
neighborhood of the point $\bar{I}$ with maximum scalar entropy
$\mathcal{S}_{g}(I|\theta)$.

Let us consider a coordinate representation
$\mathcal{R}_{\mathbf{x}}$ of the manifold $\mathcal{M}_{\theta}$
where the point $\bar{I}$ corresponds to the origin of the
coordinate frame $\mathbf{x}=0$, and the first partial derivatives
of the metric tensor vanish at the origin. Moreover, let us assume
the notation $\bar{A}=A(\mathbf{x}=0|\theta)$ to simplify the
mathematical expressions. The parametric family of geodesic
$\mathbf{x}_{g}(s|\mathbf{e})$ admits an approximate solution in
terms of power-series of the unitary affine parameter $s$:
\begin{equation}
x^{i}_{g}(s|\mathbf{e})=e^{i}(q)s-\frac{1}{6}s^{3}\partial_{l}\bar{\Gamma}^{i}_{jk}e^{j}(q)e^{k}(q)e^{l}(q)+O(s^{3}),
\end{equation}
where $\partial_{l}\bar{\Gamma}^{i}_{jk}$ is the partial derivative
of the affine connection at the origin:
\begin{equation}
\partial_{l}\bar{\Gamma}^{i}_{jk}=\frac{1}{2}\bar{g}^{im}\partial_{l}
\left[\partial_{j}\bar{g}_{mk}+\partial_{k}\bar{g}_{mj}-\partial_{m}\bar{g}_{kl}\right].
\end{equation}
Introducing the quantities $\xi^{i}_{\alpha}(q)$:
\begin{equation}
\xi^{i}_{\alpha}(q)=\frac{\partial e^{i}(q)}{\partial q^{\alpha}},
\end{equation}
the components of the projected metric tensor $g_{\alpha\beta}(r,q)$
on the boundary $\partial S^{(n)}(\bar{I},r)$ are given by:
\begin{equation}
g_{\alpha\beta}(r,q)=r^{2}\kappa_{\alpha\beta}(q)-
\frac{1}{12}r^{4}\bar{R}_{ijkl}X^{ij}_{\alpha}(q)X^{kl}_{\beta}(q)+...,
\end{equation}
where
$\kappa_{\alpha\beta}(q)=\bar{g}_{ij}\xi^{i}_{\alpha}(q)\xi^{j}_{\beta}(q)$
and $\bar{R}_{ijkl}$ is the value curvature tensor
(\ref{curvature.2}) at the origin, while the quantities
$X^{ij}_{\alpha}(q)$ are defined as:
\begin{equation}
X^{ij}_{\alpha}(q)=e^{i}(q)\xi^{j}_{\alpha}(q)-e^{j}(q)\xi^{i}_{\alpha}(q).
\end{equation}
The previous expressions lead to the following approximation for the
spherical representation of distribution function:
\begin{equation}\label{asymptotic}
dp(r,q|\theta)=\frac{1}{\mathcal{Z}_{g}(\theta)}d\varrho(r)\left[1-\frac{1}{24}r^{2}\mathcal{F}(q|\theta)+O(r^{2})\right]d\Omega(q).
\end{equation}
Here, $d\varrho(r)$ and $d\Omega(q)$ are two independent normalized
distribution functions:
\begin{eqnarray}
d\varrho(r)=\frac{1}{2^{\frac{n}{2}-1}\Gamma\left(\frac{n}{2}\right)}\exp\left(-\frac{1}{2}r^{2}\right)r^{n-1}dr,
\\
d\Omega(q)=2^{\frac{n}{2}-1} \Gamma\left(\frac{n}{2}\right)
\sqrt{\left|\frac{\kappa_{\alpha\beta}(q)}{2\pi}\right|}\frac{dq}{\sqrt{2\pi}},
\end{eqnarray}
while $\mathcal{F}(q|\theta)$ is a function on the spherical
coordinates $q$ that arises as a consequence of non-vanishing
curvature tensor $\bar{R}_{ijkl}$:
\begin{equation}\label{spherical}
\mathcal{F}(q|\theta)=\bar{R}_{ijkl}\kappa^{\alpha\beta}(q)X^{ij}_{\alpha}(q)X^{kl}_{\beta}(q).
\end{equation}
According to Eq.(\ref{asymptotic}), the curved character of the
manifold $\mathcal{M}_{\theta}$ does not allow to decouple the
radial coordinate $r$ and the spherical coordinates $q$. Such a
coupling is irreducible since the coordinate reparametrizations of
the manifold $\mathcal{M}_{\theta}$ only affect spherical
coordinates $q$ in the framework of the spherical representation of
the distribution function (\ref{polar}).
\end{proof}
\begin{corollary}
The statistical curvature tensor $R_{ijkl}(I|\theta)$ allows to
introduce some local and global invariant measures to characterize
both the intrinsic curvature of the manifold $\mathcal{M}_{\theta}$
as well as the existence of an irreducible statistical dependence
among the macroscopic variables $I$. They are the \textbf{curvature
scalar} $R(I|\theta)$ introduced in Eq.(\ref{scalar.curvature}), the
\textbf{spherical curvature scalar} $\Pi(r,q|\theta)$:
\begin{equation}
\Pi(r,q|\theta)=g^{\alpha\beta}(r,q|\theta)R_{ijkl}(r,q|\theta)X^{ij}_{\alpha}(r,q|\theta)X^{kl}_{\alpha}(r,q|\theta)
\end{equation}
with $X^{ij}_{\alpha}(r,q|\theta)$ being:
\begin{equation}
X^{ij}_{\alpha}(r,q|\theta)=\upsilon^{i}(r,q|\theta)\tau^{j}_{\alpha}(r,q|\theta)-\upsilon^{j}(r,q|\theta)\tau^{i}_{\alpha}(I|\theta),
\end{equation}
which arises as a local measure of the coupling between the radial
$r$ and the spherical coordinates $q$ in the spherical
representation of the distribution function (\ref{polar}), and
finally, the \textbf{gaussian Planck potential}
$\mathcal{P}_{g}(\theta)$ introduced in Eq.(\ref{gaussianPlanck}),
which arises as a global invariant measure of the curvature of the
manifold $\mathcal{M}_{\theta}$.
\end{corollary}
\begin{proof}
As already commented, the curvature scalar is the only invariant
associated with the first and second partial derivatives of the
metric tensor $g_{ij}(I|\theta)$. The consideration of the spherical
representation of the distribution function (\ref{polar}) allows to
introduce the normal $\upsilon^{i}(r,q|\theta)$ and tangential
vectors $\tau^{i}_{\alpha}(r,q|\theta)$, as well as the projected
metric tensor
$g_{\alpha\beta}(r,q|\theta)=g_{ij}(r,q|\theta)\tau^{i}_{\alpha}(r,q|\theta)\tau^{j}_{\beta}(r,q|\theta)$
associated with the constant scalar entropy hyper-surface $\partial
S^{(n)}(\bar{I}|r)$. This framework leads to introduce the spherical
curvature scalar $\Pi(r,q|\theta)$ as a direct generalization of the
spherical function $\mathcal{F}(q|\theta)$ of the asymptotic
distribution function (\ref{asymptotic}).

The role of the gaussian Planck potential $\mathcal{P}_{g}(\theta)$
as a global invariant measure of the curvature of the manifold
$\mathcal{M}_{\theta}$ can be easily evidenced starting from the
spherical representation of the distribution function (\ref{polar}).
Integrating over the spherical coordinates $q$, one obtains the
following expression for the gaussian partition function:
\begin{equation}\label{normalization}
\mathcal{Z}_{g}(\theta)=\frac{1}{\sqrt{2\pi}}\int^{+\infty}_{0}\exp\left(-\frac{1}{2}r^{2}\right)\Sigma_{g}(r|\theta)dr,
\end{equation}
where $\Sigma_{g}(r|\theta)$ denotes the invariant area of the
constant scalar entropy hyper-surface $\partial S^{(n)}(\bar{I}|r)$.
For a flat manifold, the invariant area $\Sigma_{flat}(r|\theta)$ is
given by:
\begin{equation}\label{flat.area}
\Sigma_{flat}(r|\theta)=\frac{\sqrt{\pi}r^{n-1}}{2^{\frac{n-1}{2}}\Gamma\left(\frac{n}{2}\right)},
\end{equation}
which is the area of an n-dimensional Euclidean sphere of radio $r$
normalized by the factor $(2\pi)^{(n-1)/2}$.
Eq.(\ref{normalization}) can be rewritten as follows:
\begin{equation}\label{Zflat}
\mathcal{Z}_{g}(\theta)=\frac{1}{\sqrt{2\pi}}\int^{+\infty}_{0}\exp\left(-\frac{1}{2}r^{2}\right)\sigma(r|\theta)\Sigma_{flat}(r|\theta)dr
\end{equation}
introducing the geometric rate $\sigma(r|\theta)$:
\begin{equation}\label{rate}
\sigma(r|\theta)=\Sigma_{g}(r|\theta)/\Sigma_{flat}(r|\theta),
\end{equation}
which characterizes how much differ the hyper-surface area of the
sphere $S^{(s)}(\bar{I},r)$ of the manifold $\mathcal{M}_{\theta}$
due to its intrinsic curvature. Since the gaussian partition
function $\mathcal{Z}_{g}(\theta)=1$ for the case of the
n-dimensional Euclidean space $\mathbb{R}^{n}$, a non-vanishing
gaussian Planck potential $\mathcal{P}_{g}(\theta)$ is a global
invariant measure of the intrinsic curvature of the manifold
$\mathcal{M}_{\theta}$.
\end{proof}
\begin{remark}
According to Eq.(\ref{estimation}) and the \textbf{corollary
\ref{entropy.planck}}, gaussian approximation (\ref{Gaussian})
employed in Ruppeiner geometry of thermodynamics \cite{Ruppeiner1}
only accounts for the local Euclidean properties of the Riemannian
manifold $\mathcal{M}_{\theta}$ at the point $\bar{I}$ with maximum
scalar entropy $\mathcal{S}_{g}\left( I|\theta \right) $.
Consequently, the intrinsic curvature of the manifold
$\mathcal{M}_{\theta}$ can be only manifested analyzing the system
fluctuating behavior beyond the Gaussian approximation
(\ref{Gaussian}).
\end{remark}
\begin{theorem}
For a sufficiently large short-range interacting thermodynamic
system, the gaussian Planck potential $\mathcal{P}_{g}(\theta)$ can
be estimated as follows:
\begin{equation}  \label{eggrenium}
\mathcal{P}_{g}(\theta)=\frac{1}{6}R(\bar{I}|\theta)+O(1/\Lambda),
\end{equation}
where $R(\bar{I}|\theta)$ is the statistical scalar curvature
(\ref{scalar.curvature}) evaluated at the point $\bar{I}$ with
maximum scalar entropy $\mathcal{S}_{g}(I|\theta)$, while $\Lambda$
is the characteristic size parameter of the system.
\end{theorem}
\begin{proof}
Assuming that scalar entropy $\mathcal{S}_{g}(I|\theta)$ and
macroscopic observables $I$ exhibits an extensive growth with the
size parameter $\Lambda$, it is easy to check the following scaling
dependencies:
\begin{equation}
\left.%
\begin{array}{c}
  \mathcal{S}_{g}\sim \Lambda \\
  I\sim \Lambda\\
\end{array}%
\right\}\Rightarrow
\left\{%
\begin{array}{c}
  g_{ij}\sim1/\Lambda,g^{ij}\sim\Lambda, \\
  \Gamma^{k}_{ij}\sim1/\Lambda,
R_{ijkl}\sim1/\Lambda^{3},R\sim1/\Lambda. \\
\end{array}%
\right.
\end{equation}
Thus, the curvature scalar $R(\bar{I}|\theta)$ constitutes a size
effect of order $1/\Lambda$ in the fluctuating behavior associated
with the distribution function (\ref{EinsteinPostulate}). According
to Riemannian geometry \cite{Berger}, the asymptotic expression of
the geometric rate $\sigma(r|\theta)$ for a small sphere can be
expressed in terms of the curvature scalar $R(\bar{I}|\theta)$ at
the center of the n-dimensional sphere $S^{(n)}(\bar{I},r)$ as
follows:
\begin{equation}\label{asymp}
\sigma(r|\theta)=1-\frac{1}{6n}R(\bar{I}|\theta)r^{2}+O(r^{2}),
\end{equation}
which leads to the desirable result (\ref{eggrenium}) performing the
integration of Eq.(\ref{Zflat}). It is worth to remark that the
estimation formula (\ref{eggrenium}) is applicable whenever the
curvature scalar $R(\bar{I}|\theta)$ be sufficiently small. This
observation suggests its possible applicability beyond the framework
of the short-range interacting systems whenever the gaussian
approximation (\ref{Gaussian}) turns a licit treatment to describe
the statistical behavior of the macroscopic observables $I$.
\end{proof}

\section{Final remarks}

As already evidenced, Riemannian formulation of classical
statistical mechanics arises as a suitable extension of Ruppeiner
geometry of thermodynamics. A main consequence of this proposal is
the interpretation of the scalar entropy $S_{g}(I|\theta)$ and its
associated thermo-statistical quantities as purely geometric notions
derived from the Riemannian structure of the manifold of macroscopic
observables $\mathcal{M}_{\theta}$. Besides, the non-vanishing of
the statistical curvature tensor $R_{ijkl}(I|\theta)$ of the
manifold $\mathcal{M}_{\theta}$ constitutes a direct indicator about
the existence of \textit{irreducible statistical dependence} for the
equilibrium distribution function $dp(I|\theta)$ associated with the
covariant redefinition of the Einstein postulate
(\ref{EinsteinPostulate}).

At first glance, thermo-statistics geometry shares some analogies
with Einstein gravitation theory. In particular, the notion of
\textit{statistical correlations} could be regarded as the
statistical counterpart of the concept of \textit{interaction}.
Thus, the relation between statistical curvature tensor
$R_{ijkl}(I|\theta)$ and the existence of irreducible statistical
dependence for the equilibrium distribution function
(\ref{EinsteinPostulate}) is analogous to the relation between the
spacetime curvature tensor $R_{ijkl}(t,\mathbf{r})$ and the
irreducible character of gravitation (while the existence of
inertial forces can be avoided with an appropriate selection of the
reference frame, gravitational forces are irreducible because of
this universal interaction is just the manifestation of the
spacetime curvature).

Before to end this section, it is worth to comment some open
questions that deserve a detailed analysis in forthcoming
contributions. The present work inspires a geometric reformulation
of classical fluctuation theory, which should lead to the derivation
of a certain set of \textit{covariant fluctuation theorems} and
\textit{dynamical equations} to describe the system relaxation.
Similarly, it is worth to analyze the implications of
thermo-statistics geometry on the study of \textit{phase transition
phenomena} as well as its possible extension to the framework of
quantum distribution functions. According to the relationship
between statistical curvature and the existence of irreducible
statistical dependence for the distribution functions, an extension
of the present proposal to quantum statistical distributions should
provide an alternative characterization for the notion of
\textit{quantum entanglement}.

\section*{Acknowledgement}

The present research has been developed with the financial support
of CONICYT/Programa Bicentenario de Ciencia y Tecnolog\'{\i}a PSD
\textbf{65} (Chilean agency).

\end{document}